\def\year{2017}\relax
\pdfoutput=1
\documentclass[letterpaper]{article}
\usepackage{aaai18}
\usepackage{times}
\usepackage{helvet}
\usepackage{courier}
\usepackage{url}
\usepackage{graphicx}
\usepackage{enumitem}
\usepackage{amsmath}
\usepackage{amssymb}
\usepackage{amsthm}
\usepackage{algorithm}
\usepackage{algorithmic}
\usepackage{multirow}

\begin{document}

\title{TIMERS: Error-Bounded SVD Restart on Dynamic Networks}

\author{Ziwei Zhang$^{1}$, Peng Cui$^{1}$, Jian Pei$^{2}$, Xiao Wang$^{1}$, Wenwu Zhu$^{1}$\\
$^{1}$ Department of Computer Science and Technology, Tsinghua University, China\\
$^{2}$ School of Computing Science, Simon Fraser University, Canada\\
\normalsize{zw-zhang16@mails.tsinghua.edu.cn, cuip@tsinghua.edu.cn, jpei@cs.sfu.ca}\\
\normalsize{wangxiao007@mail.tsinghua.edu.cn,wwzhu@tsinghua.edu.cn}\\
}

\maketitle

\newtheorem{theorem}{Theorem}
\newtheorem{lemma}{Lemma}

\begin{abstract}
    Singular Value Decomposition (SVD) is a popular approach in various network applications, such as link prediction and network parameter characterization. Incremental SVD approaches are proposed to process newly changed nodes and edges in dynamic networks. However, incremental SVD approaches suffer from serious error accumulation inevitably due to approximation on incremental updates. SVD restart is an effective approach to reset the aggregated error, but when to restart SVD for dynamic networks is not addressed in literature. In this paper, we propose TIMERS, Theoretically Instructed Maximum-Error-bounded Restart of SVD, a novel approach which optimally sets the restart time in order to reduce error accumulation in time.
    Specifically, we monitor the margin between reconstruction loss of incremental updates and the minimum loss in SVD model.
    To reduce the complexity of monitoring, we theoretically develop a lower bound of SVD minimum loss for dynamic networks and use the bound to replace the minimum loss in monitoring.
    By setting a maximum tolerated error as a threshold, we can trigger SVD restart automatically when the margin exceeds this threshold.
    We prove that the time complexity of our method is linear with respect to the number of local dynamic changes, and our method is general across different types of dynamic networks.
    We conduct extensive experiments on several synthetic and real dynamic networks. The experimental results demonstrate
    that our proposed method significantly outperforms the existing methods by reducing 27\% to 42\% in terms of the maximum error for dynamic network reconstruction when fixing the number of restarts. Our method reduces the number of restarts by 25\% to 50\% when fixing the maximum error tolerated.
\end{abstract}

\section{Introduction}
    \begin{figure*}
    \centering
    \includegraphics[width=13.6cm]{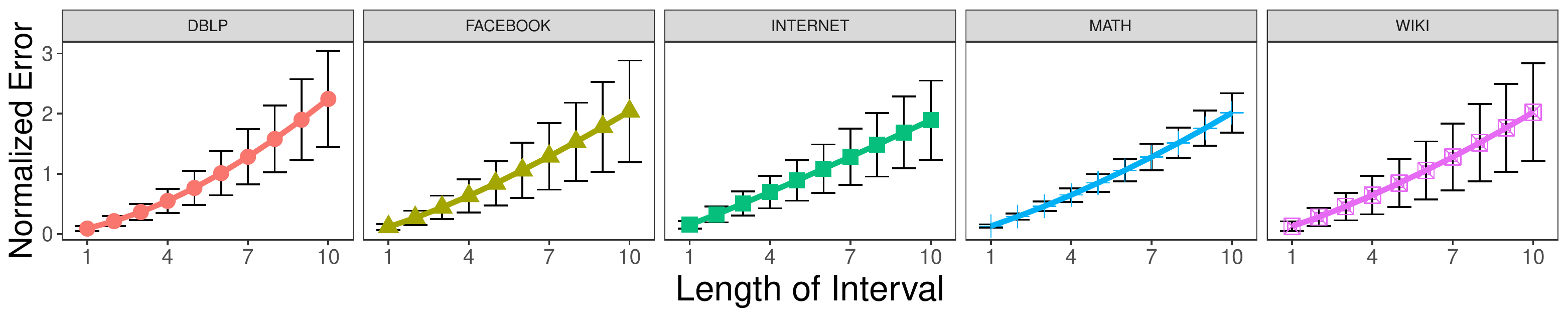}
    \caption{Preliminary study results. We calculate error accumulation in different time periods given the length of interval for calculating the error is fixed. Colored lines are the average and black error bars indicate the standard derivation. The results show that while the average error increases monotonously with the length of the interval, the variation also increases dramatically, suggesting non-uniform error accumulation in different time periods. See Section 3.1 for details.}
    \label{Preliminary2}
    \end{figure*}
    \begin{figure*}
    \centering
    \includegraphics[width=14.6cm]{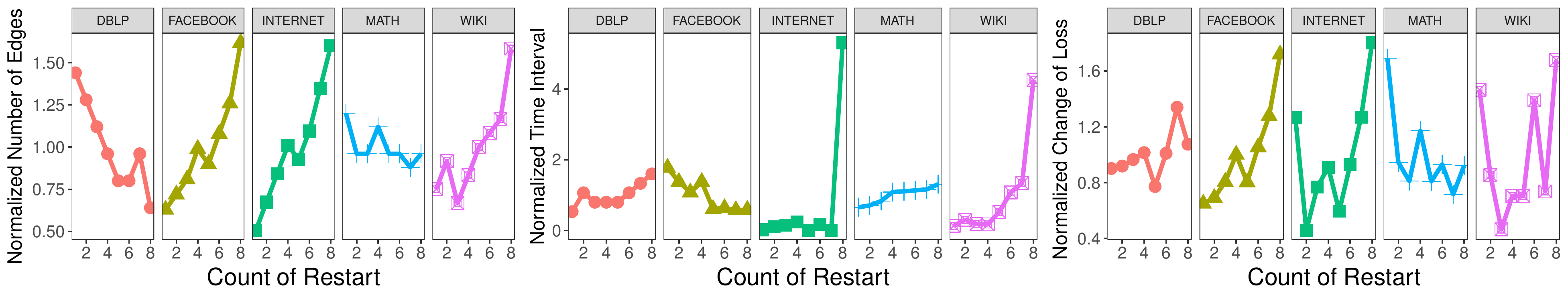}
    \caption{Preliminary study results. We reproduce the real evolving process of the network and restart SVD using the true error. Three measurements used in baselines are calculated between two consecutive restart time points: the number of edge changes, the time interval and the change of reconstruction loss. No obvious pattern is observed, suggesting that no simple correlation between the baselines and the true error can be induced. See Section 3.1 for details.}
    \label{Preliminary}
    \end{figure*}

Many methods have been proposed for network analysis, aiming to unlock the power of network data.
Among them, Singular Value Decomposition (SVD) \cite{eckart1936approximation} is a major player and proven to be successful in various network applications, such as link prediction \cite{ou2016asymmetric} and network parameter characterization \cite{chen2015fast}.

Most of the existing SVD methods, such as the widely used Lanczos Algorithm \cite{lanczos1950iteration}, are designed for static networks. In real world, however, networks are dynamic in nature. New nodes and edges may be added and existing nodes and edges may be deleted at any time. These dynamic changes make the previous SVD results unreliable. Incremental SVD methods \cite{brand2006fast,chen2015fast} are proposed to update previous SVD results to incorporate the changes without restarting the algorithm. However, as they all make some approximations in the incremental updating process, error accumulation is inevitable as the number of dynamic changes keeps increasing. In order to overcome serious deviation from the optimal SVD, incremental SVD methods still need to restart SVD at some points \cite{chen2014lwi,chen2015fast}. Surprisingly, the problem of when to restart SVD for dynamic networks largely remains open in literature. This problem is critical because the effectiveness and efficiency of SVD restarts heavily depend on the restart timing. On the one hand, a too early restart results in redundant calculation and a severe waste of computation resources. On the other hand, a too late restart leads to serious error accumulation.

How to optimally conduct SVD restart for dynamic networks is a challenging problem. The main reason is that dynamic changes of networks per se are sophisticated, and different changes have different impacts on SVD results, causing non-uniform error accumulation in different time periods or different batches of changes (as shown in Figure \ref{Preliminary2}). Therefore, this problem cannot be straightforwardly settled by heuristic methods, such as restarting after a certain time interval or a certain number of changes (as shown in Figure \ref{Preliminary}). \cite{chen2014lwi} propose to monitor the reconstruction loss and restart SVD when the loss exceeds a preset threshold. However, the reconstruction loss is constituted by both the intrinsic loss in SVD model and the accumulated error caused by approximations in incremental updates. Because the intrinsic loss in SVD model is not related to incremental updates and still exists even after restarting, it is the second part of the loss that should be used to instruct SVD restart, rather than the whole reconstruction loss.

To solve this problem, we propose TIMERS\footnote{The code is available at http://nrl.thumedia.org/} (for Theoretically Instructed Maximum-Error-bounded Restart of SVD), a novel approach to optimally set the restart time in order to reduce error accumulation in time.
Instead of monitoring the overall reconstruction loss \cite{chen2014lwi}, we propose a novel framework to monitor the margin between the reconstruction loss and the minimum loss in SVD model. We term this margin as error, which exactly measures the extra loss caused by incremental updates if we do not restart SVD. In order to avoid computing SVD minimum loss at each time slice as its complexity is the same as SVD restart, we theoretically explore a lower bound of the minimum loss based on matrix perturbation, and use the bound to calculate the error. In this way, the complexity of monitoring error is only linear with respect to the number of local dynamic changes.
Finally, we implement the maximum-error-bounded restart by setting a threshold to trigger SVD restarts automatically if the margin exceeds this threshold. We show that TIMERS has several desired merits including scalability, general applicability on different types of dynamic networks and flexibility to cooperate with current incremental SVD methods.

We demonstrate the effectiveness and efficiency of our method by conducting extensive experiments on several real and synthetic networks of different types, including weighted/unweighted, signed/unsigned ones. The results in dynamic network reconstruction demonstrate that, when fixing the number of restarts, TIMERS has a much lower maximum error (27\% to 42\%) and, when fixing the maximum error, TIMERS can significantly reduce the number of restarts (25\% to 50\%) over the existing methods. We also demonstrate that these improvements subsequently boost performance in dynamic network applications, such as link prediction and network parameter characterization.

The contributions of our paper are summarized as follows:
\begin{itemize}
\item We solve the critical open problem of when to restart SVD for dynamic networks by proposing TIMERS, a novel approach to optimally set the restart time in order to reduce error accumulation in time.
\item We theoretically explore a lower bound of SVD minimum loss on dynamic networks based on matrix perturbation, which can be efficiently calculated in linear time complexity. Based on the bound and setting a maximum tolerated error as a threshold, we implement TIMERS with theoretical foundation.
\item The experimental results demonstrate that our proposed method reduces 27\% to 42\% in terms of the maximum error when fixing the number of restarts, and reduces the number of restarts by 25\% to 50\% when fixing the maximum error tolerated.
\end{itemize}

The rest of the paper is organized as follows. In Section 2, we briefly review related works. We report observations of preliminary studies and formulate the problem in Section 3. In Section 4, we introduce TIMERS, our proposed approach. The experimental results are reported in Section 5. Finally, we conclude our findings in Section 6.

\section{Related Work}
Many methods have been proposed for network analysis \cite{ou2015non,wang2016structural,wang2017community}. Among them, Singular Value Decomposition (SVD) is proven to be successful in many important network applications. For example, the singular vectors of an adjacency matrix or its variations, such as the transition matrix, can be used to represent nodes in the network \cite{ou2016asymmetric}. Spectral analysis usually involves calculating SVD on the Laplacian matrix of a network \cite{belkin2001laplacian}. The singular values, or the equivalent eigenvalues for undirected networks, are associated with the network parameters \cite{chen2015fast}.

Incremental SVD methods have been proposed in matrix perturbation literature for decades and a brief summary of early works can be found in \cite{brand2006fast}. Basically, all early methods suffered from the problem of high time complexity or low accuracy. Later, more specific methods are proposed, for example, the algorithm to  process data with low-rank modifications \cite{brand2006fast}. Until now, incremental SVD remains an active research direction. Incremental SVD has been applied to multiple areas, such as recommendation \cite{sarwar2002incremental}, image processing \cite{brand2002incremental} and network analysis \cite{chen2015fast}. However, all methods accumulate error inevitably because some approximations have to be made in the incremental updating process \cite{chen2014lwi,chen2015fast}.

When to restart SVD, which is critical to incremental SVD, largely remains open in literature. Besides some heuristic methods that are commonly used in industry, such as restarting after a certain time interval or a certain number of changes, \cite{chen2014lwi} propose to restart SVD based on the overall reconstruction loss. As the overall construction loss includes the loss caused by the SVD model itself, this method is not suitable to set the SVD restart time.

Matrix sketching is another series of studies aiming to reduce the dimensionality of dynamic matrices \cite{liberty2013simple,cohen2015dimensionality}. The essential idea is to maintain a low dimensional ``sketch'' matrix to approximate the original matrix in certain properties. However,
 they usually require the dimensionality strictly much larger than the rank of the original matrix for theoretical guarantee \cite{cohen2015dimensionality}, which is not suitable for network data because the rank of the adjacency matrix is not rigorously low due to the power-law distribution \cite{xu2016robust}. In addition,
 most matrix sketching methods are based on the row-model (i.e. rows of the matrix come in stream) and thus cannot be applied to the setting of dynamic networks, whose elements change in arbitrary order \cite{liberty2013simple}. Our method, to the contrary, is especially designed for SVD restart problem on dynamic networks.

\section{Observations and Problem Formulation}
\subsection{Preliminary Study and Observations}\label{Baseline}
    In this section, we report some preliminary study results to investigate whether the existing approaches can solve the problem of SVD restart for dynamic networks. We summarize all currently available SVD restart methods, including two heuristic methods and one previously proposed method:
   \begin{itemize}
        \item Heu-FL: restart SVD after a fixed number of edges changed
        \item Heu-FT: restart SVD after a fixed amount of time passed
        \item LWI2 \cite{chen2014lwi}: restart SVD whenever the reconstruction loss exceeds a preset threshold
    \end{itemize}

    First, we show the fact that error accumulation is not uniform in different time periods in a dynamic network.
    Five real networks (see Section \ref{sec:dataset} for detail) are used. For each network, we divide the dynamic changes into $T=100$ time slices with an equal number of changes.
    Then, we set a starting time slice $t_0$ to run SVD and calculate error accumulation after a fixed interval $\Delta$ of time slices, i.e. restart at time slice $t_0$ and calculate the error at time slice $t_0 + \Delta$ (refer Eqn \eqref{SVDobj} for error calculation).
    We slide $t_0$ from $1$ to $T-\Delta$ to get the error in different time periods and report its mean value and standard deviation. The results of varying $\Delta$ from 1 to 10 are plotted in Figure \ref{Preliminary2}. While the average error increases monotonously with the interval length, the variation also increases dramatically, suggesting that error tends to vary greatly in different time periods.
    We get similar observations when using equal time duration in dividing time slices.

    Next, we show that the problem of setting SVD restart time for dynamic networks cannot be well solved by the aforementioned methods. Specifically, we first run SVD in each time slice as the ground truth. Then, we reproduce the real evolving process of the network by gradually adding the dynamic changes into the static network and restart SVD using the ground truth, i.e. when the true error exceeds a threshold. Three measurements are calculated between two consecutive restarts: the number of edges changed, the time interval and the change of reconstruction loss, corresponding to three baselines.
    From Figure \ref{Preliminary}, no obvious pattern can be induced for any measurement. For example, for the number of edges, three networks (FACEBOOK, INTERNET, WIKI) have a clear ascending pattern while one (DBLP) is descending and the other one (MATH) is unstable. For the time interval, outlier values appear in two networks (INTERNET, WIKI). The reconstruction loss also fails and vibrates violently.

    The above observations suggest that the existing methods cannot work well in finding the appropriate restart time for dynamic networks. The results also reveal that it is technically challenging to solve the problem because different dynamic networks have completely different dynamic patterns.
\subsection{Problem Formulation}
\subsubsection{Notations}\label{Notation}
Suppose we have an undirected network $G$ with $N$ nodes and $M$ edges. Following the commonly used notations, we use $\mathbf{A}$ to denote its adjacency matrix.
$\mathbf{A}\left(i,:\right)$ and $\mathbf{A}\left(:,i\right)$ stand for its $i^{th}$ row and column respectively.
$\mathbf{A}\left(i,j\right)$ is the weight of the edge between $i$ and $j$. For undirected networks, $\mathbf{A}$ is symmetric and $\mathbf{A}\left(i,j\right) = \mathbf{A}\left(j,i\right)$. $\mathbf{A}\left(i,j\right)$ is 0 or 1 for unweighted networks and any non-negative number for weighted networks. Negative values are allowed for signed networks. $\mathbf{A}^T$ denotes the transpose of $\mathbf{A}$ and $tr\left(\mathbf{A}\right)$ is the trace. For dynamic networks, we use subscript on the right to denote the time slice, e.g. $\mathbf{A}_{t_0}$, and the changed part of the network at time slice $t$ is denoted as $\Delta \mathbf{A}_{t}= \mathbf{A}_{t} - \mathbf{A}_{t-1}$.

\subsubsection{Maximum-Error-Bounded SVD Restart}
  SVD for networks aims to find the low rank decomposition of the similarity matrix $\mathbf{S}=\mathcal{S}\left(\mathbf{A}\right)$ by minimizing the following reconstruction loss function:
    \begin{small}
    \begin{equation}\label{SVDobj}
         \mathcal{J} =  \lVert \mathbf{S} - \mathbf{U} \cdot \mathbf{\Sigma} \cdot \mathbf{V}^T \rVert _F^2,
    \end{equation}
    \end{small}where $\mathcal{S}(\cdot)$ is a symmetric similarity function, $\mathbf{U},\mathbf{V} \in \mathbf{R}^{N \times k}$ are unitary matrices, $\mathbf{\Sigma} \in \mathbf{R}^{k \times k}$ is a non-negative diagonal matrix and $k$ is the dimensionality of the low-rank space. Some commonly used similarity functions include the high-order proximity, the Laplacian matrix, etc.
  The minimum loss, which is a function of $\mathbf{S}$ and $k$ and denoted as $\mathcal{L}\left(\mathbf{S},k\right)$, is proven to be \cite{stewart1990matrix}
  \begin{small}
    \begin{equation}\label{SVDoptObj}
        \mathcal{L}(\mathbf{S},k) = \min_{\mathbf{U},\mathbf{\Sigma},\mathbf{V}} \mathcal{J} = \sum_{l=k+1}^N \lambda_l^2,
    \end{equation}
  \end{small}where $\lambda_l$ are eigenvalues of $\mathbf{S}$ in descending magnitude.

  For dynamic networks, we have a static network $\mathbf{A}_{0}$ and its evolving data in $T$ time slices.
  The evolution of dynamic networks includes the phenomena of adding or deleting nodes and edges. As adding or deleting nodes can be simply incorporated by setting empty rows/columns in the matrix as in \cite{li2017attributed}, here we focus more on adding or deleting edges, and thus formulate the evolving process as the change of a fixed-dimension adjacency matrix $\Delta \mathbf{A}_{t},1\leq t\leq T$.

  The goal of SVD on dynamic networks is to efficiently calculate the low rank decomposition results as networks evolve, which contains two indispensable parts: incremental updates and SVD restart (i.e. recalculate optimal SVD).
  The former part usually focuses on designing a function $\mathcal{F}(\cdot)$ to efficiently update the results of previous time slice under certain approximations. Denote the change of the similarity matrix at time slice $t$ as $\Delta\mathbf{S}_t$ and we have:
  \begin{small}
    \begin{equation}
        \left[\mathbf{U}_{t},\mathbf{\Sigma}_{t},\mathbf{V}_{t}\right] =  \mathcal{F}\left(\left[\mathbf{U}_{t-1},\mathbf{\Sigma}_{t-1},\mathbf{V}_{t-1}\right],\mathbf{S}_{t-1},\Delta \mathbf{S}_{t}\right).
    \end{equation}
  \end{small} When the network evolves over time, the previous SVD results need to be updated accordingly. Although incremental SVD methods are proposed to address this issue by inducing approximations, error aggregation is inevitable in these methods, which necessitates restarting SVD at certain points.

  The most important question of SVD restart on dynamic networks is: \textbf{what are the appropriate time points}.
  Intuitively, we hope to reduce the number of restarts while keeping a low aggregated error. Here we transform the goal into a constrainted optimization problem. We set a tolerance threshold $\Theta$ on the error and then minimize the total number of restarts. Formally, we denote the error evaluation function as $\mathcal{G}(\cdot)$ and whether to restart at time slice $t$ as $c_t\in\left\{0,1\right\}$. Then the optimization objective is:
  \begin{small}
  \begin{equation}\label{MaxError}
  \begin{gathered}
        \min\limits_{c_1,...,c_{_T}}  \; \sum_{t=1}^T c_t  \\
        s.t. \; \; \mathcal{G}\left(\mathbf{S}_0...\mathbf{S}_T,\left[\mathbf{U}_{t}, \mathbf{\Sigma}_{t},\mathbf{V}_{t}\right],1\leq t \leq T\right) \leq \Theta \\
        \left[\mathbf{U}_{t},\mathbf{\Sigma}_{t},\mathbf{V}_{t}\right] = \left\{
            \begin{aligned}
                & Results \; of \; SVD  \; on \; \mathbf{S}_t  \qquad \qquad  if \;c_t = 1 \\
                & \mathcal{F}\left(\left[\mathbf{U}_{t-1},\mathbf{\Sigma}_{t-1},\mathbf{V}_{t-1}\right],\mathbf{S}_{t-1},\Delta \mathbf{S}_{t}\right)  else,
            \end{aligned}
            \right.
  \end{gathered}
  \end{equation}
  \end{small}where $\mathcal{F}(\cdot)$ is the updating function derived from incremental SVD methods.
    For the function $\mathcal{G}(\cdot)$, one approach is to directly use $\mathcal{J}(t)$, the reconstruction loss at time slice $t$. However, $\mathcal{J}(t)$ is constituted by both the minimum loss in SVD $\mathcal{L}(\mathbf{S},k)$, and the aggregated error induced by incremental updates. Since $\mathcal{L}(\mathbf{S},k)$ is intrinsic in SVD and cannot be reduced by SVD restart, $\mathcal{L}(\mathbf{S},k)$ should not be counted to guide SVD restart.
    Instead, the error induced by incremental SVD methods, i.e. the margin between the reconstruction loss and the SVD minimum loss, should be the right measure to guide SVD restart. As most applications are sensitive to the maximum error, here we define $\mathcal{G}(\cdot)$ as:
    \begin{small}
    \begin{equation}\label{NewProposeObj1}
        \mathcal{G} =  \max_{1 \leq t \leq T}   \frac{\mathcal{J}(t)-\mathcal{L}(\mathbf{S}_t,k)}{\mathcal{L}(\mathbf{S}_t,k)},
    \end{equation}
    \end{small}
    where $\mathcal{J}(t)$ is the reconstruction loss at time slice $t$. Putting Eqs. \eqref{MaxError}\eqref{NewProposeObj1} together, we have the formulation of maximum-error-bounded SVD restart on dynamic networks.
\section{TIMERS: The Proposed Method}
\subsection{The Framework for Error Monitoring}
   It is easy to see that
   \begin{small}
    \begin{equation}\label{NewProposeObj2}
        \begin{aligned}
        \max_{1 \leq t \leq T} & \frac{\mathcal{J}(t) - \mathcal{L}(\mathbf{S}_t,k)}{\mathcal{L}(\mathbf{S}_t,k)} \leq \Theta \\
            \Leftrightarrow & \frac{\mathcal{J}(t) - \mathcal{L}(\mathbf{S}_t,k)}{\mathcal{L}(\mathbf{S}_t,k)} \leq \Theta \quad \forall 1 \leq t \leq T.
        \end{aligned}
    \end{equation}
    \end{small}To ensure the above equation satisfied, one straightforward way is to monitor the error and restart whenever the error exceeds $\Theta$.
    Intuitively, the total number of restarts is reduced because we only restart when the error reaches the threshold.
    However, this is still problematic because directly calculating $\mathcal{L}(\mathbf{S}_t,k)$ has the same time complexity as SVD restart. Alternatively, if we can find a lower bound $B(t)> 0$ so that
    \begin{small}
    \begin{equation}\label{ProposeRelex}
        \mathcal{L}(\mathbf{S}_t,k) \geq B(t) \Rightarrow \frac{\mathcal{J}(t) - \mathcal{L}(\mathbf{S}_t,k)}{\mathcal{L}(\mathbf{S}_t,k)} \leq \frac{\mathcal{J}(t) - B(t)}{B(t)}.
    \end{equation}
     \end{small}Then, we can relax Eq. \eqref{NewProposeObj2} to
    \begin{small}
    \begin{equation}\label{NewProposeObj3}
        \frac{\mathcal{J}(t) - B(t)}{B(t)} \leq \Theta \quad \forall 1 \leq t \leq T.
    \end{equation}
    \end{small}After that, we can monitor Eq. \eqref{NewProposeObj3} instead and restart SVD whenever it is not satisfied.
    The remaining question is how to find $B(t)$ that can be efficiently calculated.

\subsection{A Lower Bound of SVD Minimum Loss}
    Next, we derive the lower bound based on matrix perturbation.
    \begin{theorem}[A Lower Bound of SVD Minimum Loss]
        If $\mathbf{S}$ and $\Delta \mathbf{S}$ are symmetric matrices, then:
        \begin{small}
        \begin{equation}\label{PertInequ}
         \mathcal{L}(\mathbf{S} + \Delta \mathbf{S},k) \geq \mathcal{L}(\mathbf{S},k) + \Delta tr^2(\mathbf{S}+\Delta \mathbf{S},\mathbf{S}) - \sum_{l=1}^k \lambda_l,
        \end{equation}
        \end{small}where $\lambda_1\geq \lambda_2...\geq\lambda_k$ are the top-$k$ eigenvalues of $\nabla_{S^2} = \mathbf{S}\cdot\Delta\mathbf{S} +\Delta\mathbf{S}\cdot\mathbf{S} + \Delta\mathbf{ S}\cdot \Delta\mathbf{S} $, and
        \begin{small}
        $$\Delta tr^2(\mathbf{S}+\Delta \mathbf{S},\mathbf{S}) = tr\left( (\mathbf{S}+\Delta \mathbf{S})\cdot(\mathbf{S}+\Delta \mathbf{S}) \right) - tr(\mathbf{S}\cdot\mathbf{S}).$$
        \end{small}
    \end{theorem}
    \begin{proof}
        Denote $\Lambda^{k}(\mathbf{Q})$ as the sum of the top-k eigenvalues for any symmetric matrix $\mathbf{Q}$, i.e.
        \begin{small}
        \begin{equation}\label{Eigensum}
            \Lambda^{k}(\mathbf{Q}) = \sum_{l=1}^k \lambda_l^\prime,
        \end{equation}
        \end{small}where $\lambda_1^\prime \geq \lambda_2^\prime ... \geq \lambda_k^\prime$ are the top-k eigenvalues of $\mathbf{Q}$. Then, we can rewrite $\mathcal{L}(\mathbf{S},k)$, defined in Eq. \eqref{SVDoptObj}, as:
        \begin{small}
        \begin{equation}\label{SVDoptObj2}
        \mathcal{L}(\mathbf{S},k) = \sum_{l=1}^{N} \lambda_l^2 - \sum_{l=1}^{k} \lambda_l^2 = tr(\mathbf{S} \cdot \mathbf{S}^T) - \Lambda^{k}(\mathbf{S}\cdot \mathbf{S}^T).
        \end{equation}
        \end{small}According to matrix perturbation theory \cite{stewart1990matrix}, for any two symmetric matrices $\mathbf{P}$ and $\mathbf{Q}$, we have
        \begin{small}
        \begin{equation}\label{EigensumPertub}
            \Lambda^{k}(\mathbf{P + Q}) \leq \Lambda^{k}(\mathbf{P}) + \Lambda^{k}(\mathbf{Q}).
        \end{equation}
        \end{small}Let $\mathbf{P} = \mathbf{S} \cdot \mathbf{S},\; \mathbf{P + Q} = (\mathbf{S} + \Delta \mathbf{S}) \cdot (\mathbf{S} + \Delta \mathbf{S})$. By putting Eq. \eqref{SVDoptObj2} and Eq. \eqref{EigensumPertub} together, we finish the proof.
    \end{proof}
    The theorem shows if we have calculated the minimum loss for $\mathbf{S}$, without calculating SVD on $\mathbf{S}+\Delta \mathbf{S}$, matrix perturbation can lead to a lower bound on the new minimum loss by treating $\Delta \mathbf{S}$ as a perturbation of the original similarity $\mathbf{S}$. From the theorem, we can set:
    \begin{small}
    \begin{equation}\label{BoundCal2}
    \begin{split}
        B(t) =  & \mathcal{L}(\mathbf{S}_{t^\prime},k) + \Delta tr^2(\mathbf{S}_t,\mathbf{S}_{t^\prime}) - \Lambda^{k}(\nabla_{S^2}),
    \end{split}
    \end{equation}
    \end{small}where $\Delta\mathbf{S} = \mathbf{S}_t - \mathbf{S}_{t^\prime},\nabla_{S^2}=\mathbf{S}_{t^\prime}\cdot\Delta\mathbf{S}  +\Delta\mathbf{S}\cdot\mathbf{S}_{t^\prime}+ \Delta\mathbf{S}\cdot\Delta\mathbf{S}$ and $t^\prime$ is the last time when we calculated SVD.

     As discussed earlier, efficiency is a key issue. Here we analyze the time complexity of calculating the bound. First, we have the following two lemmas (the proofs are omitted due to limited space):
    \begin{lemma}
        The time complexity of calculating $\mathcal{L}(\mathbf{S},k)$ after obtaining the optimal $\mathbf{U},\mathbf{\Sigma},\mathbf{V}$ in SVD is $O(M+k)$, where $M$ in the number of non-zero elements in $\mathbf{S}$.
    \end{lemma}
    \begin{lemma}
        Calculating the change of the reconstruction loss has the following time complexities:
        \begin{itemize}[leftmargin=0.2cm]
        \item $O(kM_S)$, if the similarity matrix changes $\Delta \mathbf{S}$ and $M_S$ is the number of non-zero elements in $\Delta \mathbf{S}$;
        \item $O(kd_i + k^2)$, if vectors of node $i$ change and $d_i$ is its degree.
        \end{itemize}
    \end{lemma}
    Lemma 1 shows it takes little additional time to compute the minimum loss if we have calculated the results of SVD. Lemma 2 shows we can efficiently compute the change of the reconstruction loss when the network evolves and the SVD results are incrementally updated. These lemmas ensure calculating $\mathcal{J}(t)$ is not time-consuming. Then, we analyze the time complexity of calculating $B(t)$.
    \begin{theorem}
        The time complexity of calculating $B(t)$ in Eq. \eqref{BoundCal2} is $O(M_S + M_L k + N_L k^2)$, where $M_S$ is the number of the non-zero elements in $\Delta \mathbf{S}$, and $N_L,M_L$ are the number of the non-zero rows and elements in $\nabla_{S^2}$ respectively.
    \end{theorem}
    \begin{proof}
        Because all matrices are symmetric, we have
        \begin{small}
        $$\Delta tr^2(\mathbf{S}+\Delta \mathbf{S},\mathbf{S}) = \sum_{\Delta\mathbf{S}(i,j)\neq 0} \left[(\mathbf{S}+\Delta \mathbf{S})(i,j)^2 - \mathbf{S}(i,j)^2 \right].$$
        \end{small}Then, we can use numeric methods, such as the Lanczos algorithm \cite{lanczos1950iteration}, to calculate the top-$k$ eigenvalues of $\nabla_{S^2}$, which is known to be $O(M_L k + N_L k^2)$. Putting them together leads to the result.
    \end{proof}
    This time complexity includes two parts. The first part is equal to the number of changed elements in the similarity matrix. The second part shows that, instead of calculating the top-$k$ eigenvalues of $\mathbf{S}+\Delta \mathbf{S}$, we only need to calculate the top-$k$ eigenvalues of $\nabla_{S^2}$ which has the same scale of non-zero elements as $\mathbf{S} \cdot \Delta \mathbf{S}$. When the changed parts only occupy a tiny fraction of the whole network, $\mathbf{S} \cdot \Delta \mathbf{S}$ will be local and have much fewer non-zero elements than $\mathbf{S}+\Delta \mathbf{S}$. To make it more clear, we calculate the exact results for two typical types of networks:
    \begin{itemize}[leftmargin = 0.3cm]
        \item If every node has a equal probability of adding new edges, we have: $M_L \approx 2 d_{avg}M_S$, where $d_{avg}$ is the average degree of the network .
        \item For Barabasi Albert model \cite{barabasi1999emergence}, a typical example of preferential attachment networks, we have: $M_L \approx \frac{12}{\pi^2}\left[log(d_{max})+\gamma\right]M_S$, where $d_{max}$ is the maximum degree of the network and $\gamma \approx 0.6$ is a constant.
    \end{itemize}
    In short, the complexity of calculating $B(t)$ is only related to the local dynamic changes but does not involve the whole network for calculation.

    In these theorems, we only require the matrices to be symmetric and sparse. Therefore, the lower bound is general across different types of networks (e.g. weighted or unweighted, signed or unsigned), and different dynamic scenarios (e.g. add or delete edges, adjust edge weights).

    Algorithm \ref{TIR} shows the overall method. It is straightforward to see that our method has no specific requirement on the updating method, i.e. TIMERS is flexible to cooperate with any incremental SVD method. The overall time complexity includes the piecewise linear complexities between restarts and the complexity of SVD restarts.
    \begin{algorithm}[t]
    \caption{TIMERS: $\mathbf{T}$heoretically $\mathbf{I}$nstructed $\mathbf{M}$aximum $\mathbf{E}$rror-bounded $\mathbf{R}$estart of $\mathbf{S}$VD}
    \label{TIR}
    \begin{algorithmic}[1]
    \REQUIRE Static Adjacency Matrix $\mathbf{A}_0$, Dynamic Changes $\Delta \mathbf{A}_1...\Delta \mathbf{A}_T$, Similarity Function $\mathcal{S}(\cdot)$, Dimensionality $k$, Error Threshold $\Theta$, Incremental SVD method $\mathcal{F}(\cdot)$
    \ENSURE  SVD results in each time slice $\left[ \mathbf{U}_t,\mathbf{\Sigma}_t,\mathbf{V}_t \right]$
    \STATE Calculate the initial similarity $\mathbf{S}_0 = \mathcal{S}(\mathbf{A}_0)$
    \STATE Calculate SVD on $\mathbf{S}_0$ to obtain $\left[ \mathbf{U}_0,\mathbf{\Sigma}_0,\mathbf{V}_0 \right]$
    \FOR{t in 1:T}
        \STATE Calculate the similarity change in that time slice $\Delta \mathbf{S}_t$
        \STATE Use $\mathcal{F}(\cdot)$ to get updated SVD results $\left[ \mathbf{U}_t,\mathbf{\Sigma}_t,\mathbf{V}_t \right]$
        \STATE Compute loss $\mathcal{J}(t)$ and bound $B(t)$ using \eqref{BoundCal2}
        \IF{ $\frac{\mathcal{J}(t) - B(t)}{B(t)} > \Theta$}
          \STATE Calculate SVD on $\mathbf{S}_t$ to obtain new $\left[\mathbf{U}_t,\mathbf{\Sigma}_t,\mathbf{V}_t \right]$
        \ENDIF
        \STATE Return $\left[\mathbf{U}_t,\mathbf{\Sigma}_t,\mathbf{V}_t \right]$
    \ENDFOR
    \end{algorithmic}
    \end{algorithm}

\section{Experiments}
In this section, we conduct extensive experiments to evaluate our method and demonstrate its advantages through comparative study. Note that the baselines have been introduced in Section 3.1.
\subsection{Datasets}\label{sec:dataset}
    To comprehensively evaluate the effectiveness of TIMERS, we apply it to five real dynamic networks\footnote{All networks are publicly available at http://snap.stanford.edu/ or http://konect.uni-koblenz.de/} from different domains: three social networks, one co-author network and one Internet topology network. These networks are of different types including weighted or unweighted, signed or unsigned ones, and all of them have real timestamps.
    \begin{itemize}[leftmargin=0.27cm]
    \item FACEBOOK, MATH, WIKI: they are online social networks in Facebook, MathOverflow and Wikipedia. A node represents a user and an edge represents the social link between two users. The edges are unweighted and unsigned in FACEBOOK, weighted in MATH and signed in WIKI. Negative edges exist for representing conflicts between users, e.g. edit-wars.
    \item DBLP is a collaboration network of computer science researchers. Since authors can have multiple common publications, the network is weighted.
    \item INTERNET is a network of autonomous system connections. The edges are weighted to indicate the number of connections between two systems.
    \end{itemize}
   To compare the effectiveness of different methods, SVD has to be conducted in each time slice as the ground truth, which is very time consuming. Therefore, we sample subsets of the original networks and the sizes of the sampled networks are listed in Table \ref{Datasets}. For each network, we divide the edges into static and evolving part according to their timestamp order. We further divide the evolving edges into 50 time slices with an equal number of edges to simulate its evolving process.

    To get insightful understanding, we also experiment on three synthetic networks whose structures are controllable. We simulate some important dynamic characteristics of real networks with the following methods:
    \begin{itemize}[leftmargin=0.4cm]
        \item RANDOM is generated by random graph model \cite{erdos1960evolution} where nodes form edges randomly.
        \item RANDOM-Cel is a variant of RANDOM by simulating the appearance of celebrities in online social networks. Specifically, we randomly select a time and a node as the celebrity. Then, we randomly sample a proportion of all nodes that connect to the celebrity node, i.e. simulating that someone suddenly gets the attention of many others.
        \item RANDOM-Com is another variant of RANDOM where community structures suddenly form, which often happens in real world triggered by off-line events. Specifically, we randomly select a time and assign a proportion of nodes to some communities, within which the nodes will have a high probability to form edges. The model is tuned by three parameters: the proportion of nodes, the number of communities and the edge forming probability.
    \end{itemize}
    We tune the parameters to control the percentage of edges generated by celebrities and communities to simulate different degrees of changes in the network structures.
    \begin{table}
    \small
    \caption{The Statistics of Datasets}\label{Datasets}
    \begin{tabular}{ l | c | c | c | c}
    \hline
    Dataset & Nodes & Static E & Evolving E & Type \\ \hline
    FACEBOOK & 52804 & 962654 & 632188 & SN\\ \hline
    MATH & 13586 & 600000 & 138408& W\\ \hline
    WIKI & 28223 & 1000000 & 375270& W,S\\ \hline
    DBLP & 28331 & 200000 & 79436 & W\\ \hline
    INTERNET & 32077 & 111644 & 196270 & W,SN\\ \hline
    RANDOM & 5000  & 60000 & 210000 &  - \\ \hline
    RANDOM-Cel & 5000 & 60000 & 210000 & -\\ \hline
    RANDOM-Com & 5000 & 60000 & 210000 & -\\ \hline
    \end{tabular}
    \\
    W = Weighted, S = Signed, SN = Static Network marked
    \end{table}

    \begin{table*}
    \small
    \caption{Dynamic network reconstruction. Relative error when fixing the number of restarts.}\label{LossValues}
    \centering
    \begin{tabular}{| l | c | c | c | c | c | c | c | c | c |}
    \hline
    \multirow{2}{*}{Dataset} & \multicolumn{4}{c|}{$avg(r)$} & \multicolumn{4}{c|}{$max(r)$}\\
      \cline{2-9}         &  TIMERS & LWI2 & Heu-FL & Heu-FT & TIMERS & LWI2 &  Heu-FL & Heu-FT \\ \hline
    FACEBOOK  & $\mathbf{0.005}$ & 0.020 &0.009 & 0.011 & $\mathbf{0.014}$ & 0.038 &0.025 & 0.023 \\ \hline
    MATH      & $\mathbf{0.037}$ & 0.057 &0.044 & 0.051 & $\mathbf{0.085}$ & 0.226 &0.117 & 0.179 \\ \hline
    WIKI      & $\mathbf{0.053}$ & 0.086 &0.071 & 0.281 & $\mathbf{0.139}$ & 0.332 &0.240 & 0.825 \\ \hline
    DBLP      & $\mathbf{0.042}$ & 0.110 &0.053 & 0.064 & $\mathbf{0.121}$ & 0.386 &0.198 & 0.238 \\ \hline
    INTERNET  & $\mathbf{0.152}$ & 0.218 &0.196 & 0.961 & $\mathbf{0.385}$ & 0.806 &0.647 & 1.897 \\ \hline
    \end{tabular}
    \\
    The best results are marked bold for each dataset in both measurements.
    \end{table*}
    \begin{table*}
    \small
    \begin{minipage}[b]{0.47\linewidth}
    \caption{Link Prediction Relative Error of MSE (\%)}\label{LinkPrediction}
    \centering
    \begin{tabular}{| l | c | c | c | c | c |}
    \hline
    Dataset    & TIMERS & LWI2 & Heu-FL & Heu-FT \\ \hline
    FACEBOOK   & $\mathbf{1.54}^*$ & 4.27 & 2.21 & 2.81  \\ \hline
    MATH       & $\mathbf{1.14}^*$ & 2.68 & 1.31 & 1.29  \\ \hline
    WIKI       & $\mathbf{1.06}^*$ & 3.44 & 1.63 & 4.13  \\ \hline
    DBLP       & $\mathbf{0.18}^*$ & 0.32 & 0.22 & 0.27  \\ \hline
    INTERNET   & $\mathbf{11.27}^*$ & 23.36 & 16.98 & 34.30 \\ \hline
    \end{tabular}
    \\
    *: outperform other methods at 0.005 level paired t-test in 10 runs.
    \end{minipage}
    \begin{minipage}[b]{0.52\linewidth}
      \caption{First Eigenvalue Tracking Error in RMSE}\label{EigenTrack}
    \centering
    \begin{tabular}{| l | c | c | c | c | c|}
    \hline
    Dataset & TIMERS & LWI2 & Heu-FL & Heu-FT \\ \hline
    FACEBOOK  & $\mathbf{0.66}$ & 0.94 &1.11 & 1.53   \\ \hline
    MATH      & $\mathbf{2.27}$ & 5.03 &4.92 & 4.80   \\ \hline
    WIKI      & $\mathbf{15.45}$ & 18.42 &17.73 & 97.08   \\ \hline
    DBLP      & $\mathbf{8.31}$ & 11.42 &13.92 & 24.66   \\ \hline
    INTERNET  & $\mathbf{4.18}$ & 12.56 &7.95 & 57.38   \\ \hline
    \end{tabular}
    \\
    The best result is marked bold for each dataset.
    \end{minipage}
    \end{table*}
    The statistics of all datasets are summarized in Table \ref{Datasets}. All experiments are conducted in a single PC with 2 i7-6700 CPU and 24GB memory in MATLAB language.

\subsection{Dynamic Network Reconstruction}
\subsubsection{Experimental Setting}
    The primal objective of SVD on dynamic networks is to reconstruct the given similarity matrix at each time slice. First, we validate different methods in terms of reconstruction. The procedure of reconstruction is as follows. At first, all methods calculate SVD on the static network to get the initial results. After that, the evolving edges of the network come in time slices and different methods decide whether to restart SVD individually. The loss at time slice $t$ can be calculated using Eq. \eqref{SVDobj}. Specifically, in our experiment, we set the similarity matrix to be the adjacency matrix for simplicity, and $k$ to be 100 as commonly used. To purely compare the effectiveness of the restart time, no incremental updating between time slices is adopted unless stated otherwise. For the evaluation metric, we use relative error defined as:
    \begin{small}
    \begin{equation}\label{eq:relative_error}
                r_t = \frac{ Loss \; at \; t -  Minimum \; Loss \; at\; t}{Minimum \; Loss \; at \;t} = \frac{\mathcal{J}(t) - \mathcal{L}(\mathbf{S}_t,k) }{ \mathcal{L}(\mathbf{S}_t,k)}.
    \end{equation}
    \end{small}$\mathcal{L}(\mathbf{S}_t,k)$ is defined in Eq. \eqref{SVDoptObj} and calculated by SVD as ground truth. We further take two measurements: the maximum error over all time slices $max(r)= \max_{1\leq t\leq T} r_t$ and the average error $avg(r)=\frac{1}{T} \sum_{t=1}^T r_t$.
\subsubsection{Fixing the Number of Restarts}
First, we report the results when fixing the number of restarts for all methods. Specifically, we directly set the number of restarts for two heuristic methods. For our method and LWI2, we adjust the threshold on error, so that all methods have the same number of restarts.
From Table \ref{LossValues}, we can observe that our method outperforms all the baselines on all dynamic networks. For example, in the largest network FACEBOOK, TIMERS can reduce the average error by 41.2\% and the maximum error by 39.8\%. These results demonstrate that the timing of SVD restart is indeed crucial for dynamic networks, and our method has better performance capturing the appropriate restart timing than heuristic methods. LWI2 fails because the overall reconstruction loss is not a good measurement.

\begin{figure}
\includegraphics[width=8.2cm]{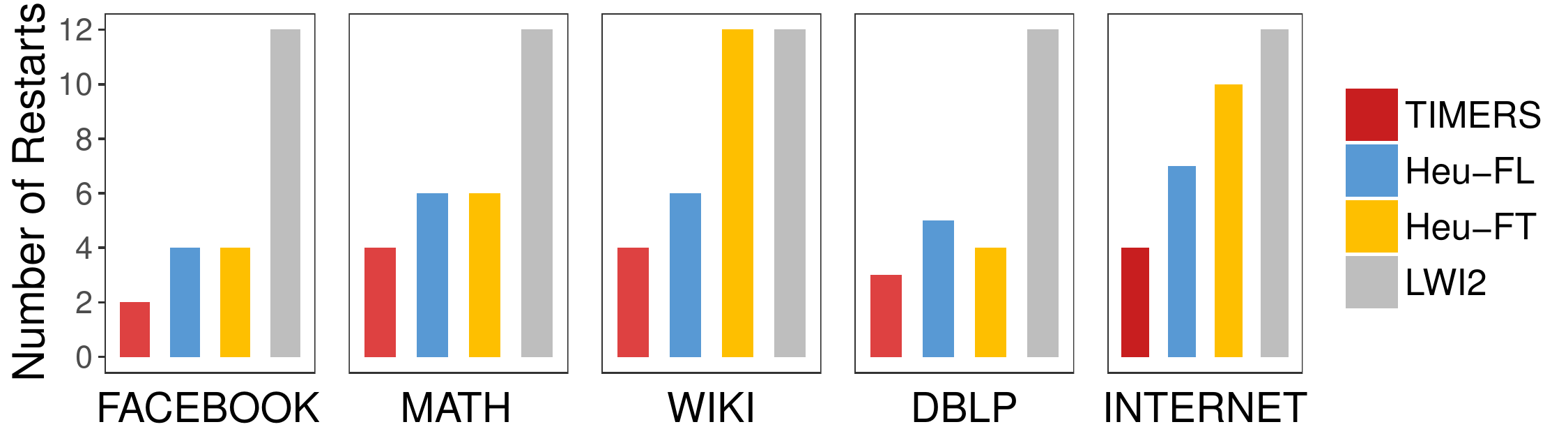}
\caption{Dynamic network reconstruction. The number of restarts needed when fixing the same maximum error.}
\label{Reconstruction}
\end{figure}

\subsubsection{Fixing the Maximum Error}
We further report the number of restarts of each method when the maximum error is fixed, i.e. we control the threshold in TIMERS and LWI2, and adjust the number of restarts for heuristic methods, so that all methods achieve the same maximum error. As shown in Figure \ref{Reconstruction}, TIMERS greatly reduces the number of restarts while maintaining the same maximum error. In FACEBOOK, the reducing rate achieves 50\%. This demonstrates that our method could save lots of computation resources while maintaining similar SVD accuracy as other methods.

\subsection{Dynamic Network Applications}
The effectiveness of TIMERS in reconstruction lays the foundation of its gains in applications of dynamic networks. Next, we conduct experiments on two typical applications: link prediction and network parameter characterization. The former focuses on predicting individual edges while the latter focuses on macroscopical indexes of the network.

\subsubsection{Link Prediction}
Link Prediction is an important dynamic network application using SVD. Specifically, we randomly hide 10\% of the network and test whether SVD on the rest of the network can recover them. For the evaluation metric, because both weighted and signed networks exist and the dynamic changes include adding new edges and changing edge weights, some standard metrics in link prediction, such as precision or AUC, do not fit in our experiment. Here, we use Mean Square Error (MSE) \cite{levinson1946wiener} as a replacement. In addition, what we aim to evaluate is how well different methods can approximate the performance of the optimal SVD. So we calculate the relative error in a similar way as Eq. \eqref{eq:relative_error} using MSE. We report the average results of 10 runs in Table \ref{LinkPrediction}. We can see that TIMERS consistently outperforms baseline methods.

\subsubsection{Network Parameter Characterization}
Some important network parameters can be characterized by the top-$k$ eigenvalues of the network. When networks evolve over time, these parameters need to be tracked. Here, we choose one state-of-the-art incremental SVD method Trip \cite{chen2015fast} to incrementally update SVD results between two time slices. Following their work, the largest eigenvalue of the adjacency matrix is selected as the target value and rooted mean square error (RMSE) \cite{levinson1946wiener} is adopted as the measurement. For fair comparison, we fix the number of restarts for all methods. From Table \ref{EigenTrack}, we can see that TIMERS achieves the best results on all networks, demonstrating the effectiveness of our proposed method. On MATH and INTERNET, we can reduce RMSE by 50\%.
\subsection{Analysis}\label{AnalysisPart}
\subsubsection{Robustness Analysis}
Next, we conduct experiments in synthetic datasets to analyze whether TIMERS is robust and under what circumstances can it gain more. We vary the percentage of edges generated by celebrities and communities to simulate different network evolving scenarios.
\begin{figure}
\centering
\includegraphics[width=7.75cm]{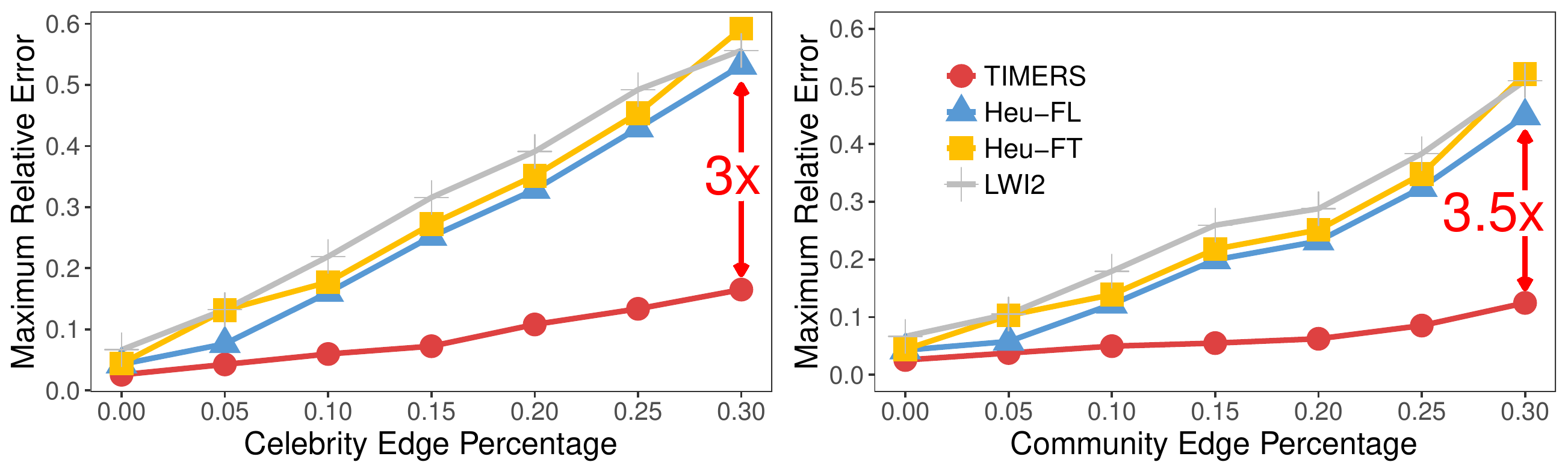}
\caption{Robustness analysis. Reconstruction error on synthetic networks with varying network structures.}
\label{LossValues2}
\end{figure}
From Figure \ref{LossValues2}, we can observe that the improvement of TIMERS has a positive correlation with the percentage of celebrity and community edges. When 30\% of the new edges are caused by celebrities or communities, TIMERS reports more than 3 times improvement in terms of the maximum error. Similar results are observed in link prediction and network parameter characterization tasks but are omitted for brevity.

The generation of celebrities or communities often lays dramatic influence on the network structure and thus SVD restart is more demanded to capture such changes. The experimental results show TIMERS is able to detect this change and then instructs the algorithm to restart in time. It also suggests TIMERS is robust and able to prevent error accumulation in unusual situations of network evolution.

\begin{figure}
\centering
\includegraphics[width=7.5cm]{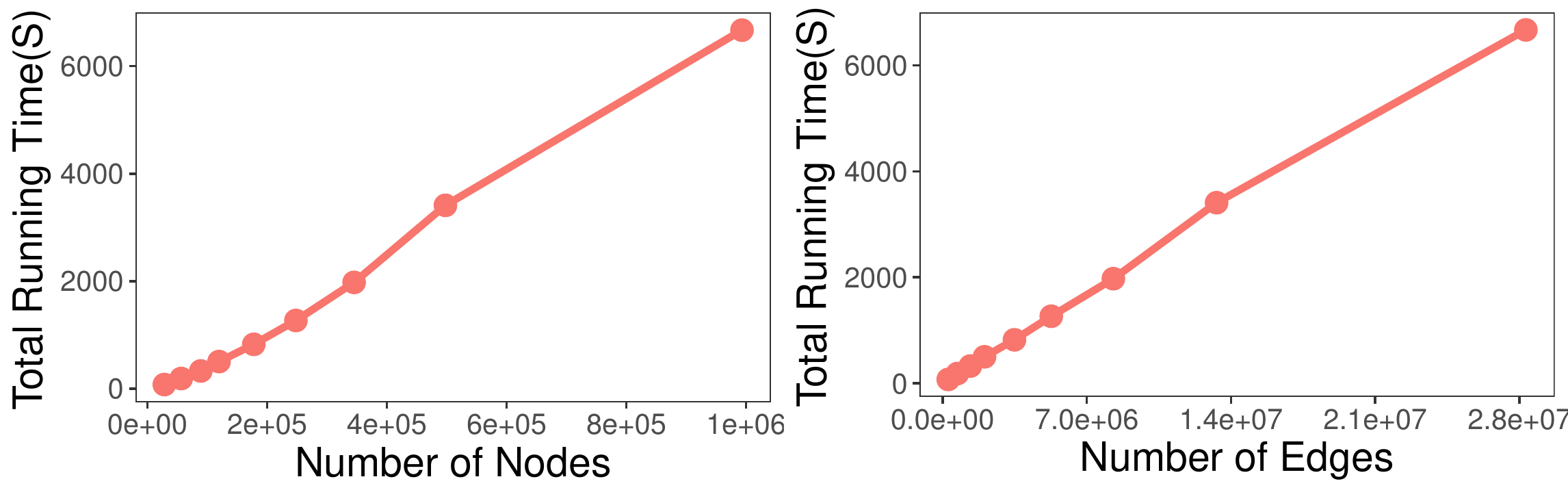}
\caption{Scalability analysis. The total running time of TIMERS w.r.t. the number of nodes and edges on DBLP.}
\label{Scala}
\end{figure}

\subsubsection{Scalability Analysis}\label{Scalability}
Now, we analyze the scalability of TIMERS. We use the same experimental setting and record the running time when the size of the network grows. Results of DBLP, whose original size is the largest among our datasets, are reported in Figure \ref{Scala}. The running time grows linearly with the network size (number of nodes and number of edges respectively), showing that our method is scalable.
\subsubsection{Parameter Analysis}
In our experiments, we fix the dimensionality of the low-rank space as $k$=100, which is widely used in the existing works. Similar results are observed for different $k$ but are omitted for the lack of space. The other important parameter is the error threshold $\Theta$. Qualitatively, larger $\Theta$ will tolerate more error and leads to fewer number of restarts. The setting of $\Theta$ depends on application scenarios. How to rigorously set $\Theta$ to precisely control the number of restart is left as future work.

\section{Conclusion}
In this paper, we tackle the problem of SVD restart time for dynamic networks and propose TIMERS, a novel approach based on monitoring and bounding the maximum error. By exploring a lower bound of the SVD minimum loss on dynamic networks, we can trigger SVD restart automatically when the margin between the reconstruction loss and the lower bound exceeds a preset threshold. We show that our method is scalable and general across different types of networks. Extensive experimental results on synthetic and real dynamic networks show that TIMERS outperforms the existing methods in all tasks. One future direction is to generalize this idea to directed networks and non-square matrices.

\section{Acknowledgments}
This work was supported in part by National Program on Key Basic Research Project (No. 2015CB352300), National Natural Science Foundation of China (No. 61772304, No. 61521002, No. 61531006, No. 61702296), National Natural Science Foundation of China Major Project (No. U1611461), the NSERC Discovery Grant program, the Canada Research Chair program, the NSERC Strategic Grant program, the research fund of Tsinghua-Tencent Joint Laboratory for Internet Innovation Technology, and the Young Elite Scientist Sponsorship Program by CAST. Peng Cui, Xiao Wang and Wenwu Zhu are the corresponding authors. All opinions, findings, conclusions and recommendations in this paper are those of the authors and do not necessarily reflect the views of the funding agencies.

\fontsize{9pt}{10pt}
\selectfont
\bibliography{ms}
\bibliographystyle{aaai}

\end{document}